\documentclass[aps,pra,10pt,twocolumn]{revtex4-2}
\usepackage{amsfonts,amsmath,amssymb,amsthm}
\usepackage{braket}
\usepackage{graphicx}
\usepackage{xspace}

\newcommand{\QSL}{QSL\xspace}
\newcommand{\ML}{ML\xspace}

\DeclareMathOperator{\esssup}{ess\ sup}

\newtheorem{lemma}{Lemma}
\newtheorem{theorem}{Theorem}
\newtheorem{corollary}{Corollary}
\newtheorem{remark}{Remark}

\begin{document}
\title{$\boldsymbol{\alpha_{>}(\epsilon) = \alpha_{<}(\epsilon)}$ For The
 Margolus-Levitin Quantum Speed Limit Bound}
\author{H. F. Chau}
\email{{\tt hfchau@hku.hk}}
\affiliation{Department of Physics, University of Hong Kong, Pokfulam Road,
 Hong Kong}
\date{\today}

\begin{abstract}
 The Margolus-Levitin (\ML) bound says that for any time-independent
 Hamiltonian, the time needed to evolve from one quantum state to another is
 at least $\pi \alpha(\epsilon) / (2 \braket{E-E_0})$, where $\braket{E-E_0}$
 is the expected energy of the system relative to the ground state of the
 Hamiltonian and $\alpha(\epsilon)$ is a function of the fidelity $\epsilon$
 between the two state.  For a long time, only a upper bound
 $\alpha_{>}(\epsilon)$ and lower bound $\alpha_{<}(\epsilon)$ are known
 although they agree up to at least seven significant figures.  Lately,
 H\"{o}rnedal and S\"{o}nnerborn proved an analytical expression for
 $\alpha(\epsilon)$, fully classified systems whose evolution times saturate
 the \ML bound, and gave this bound a symplectic-geometric interpretation.
 Here I solve the same problem through an elementary proof of the \ML bound.
 By explicitly finding all the states that saturate the \ML bound, I show that
 $\alpha_{>}(\epsilon)$ is indeed equal to $\alpha_{<}(\epsilon)$.  More
 importantly, I point out a numerical stability issue in computing
 $\alpha_{>}(\epsilon)$ and report a simple way to evaluate it efficiently and
 accurately.
\end{abstract}

\maketitle

\section{Introduction}
\label{Sec:Intro}
 Quantum information processing speed cannot be arbitrarily fast.  In
 particular, Margolus and Levitin proved that to evolve from a state to
 another state orthogonal to it through a time-independent Hamiltonian, the 
 minimum time required is inversely proportional to its expected energy
 relative to the ground state of the Hamiltonian~\cite{PHYSCOMP96,ML98}.  This
 so-called \ML bound is a significant result for it means that non-zero
 evolution time is required to change a quantum state under any
 time-independent Hamiltonian.  Time, therefore, is a genuine resource in
 quantum information processing.  Since then, a lot of bounds of this type,
 commonly known as quantum speed limits (\QSL{s}), are
 found~\cite{[{See, for example, }][{ and references cited therein.}]Frey16}.

 Shortly after the discovery of this \ML bound, Giovannetti
 \emph{et al.}~\cite{GLM03} extended it to the more general situation.  More
 precisely, they proposed that the evolution time $\tau$ from a state $\rho$
 to another state $\rho'$ under any time-independent Hamiltonian must be
 lower-bounded by
\begin{equation}
 \frac{\tau}{\hbar} \ge \frac{\pi\alpha(\epsilon)}{2\braket{E-E_0}} ,
 \label{E:ML_bound_original_form}
\end{equation}
 where $\epsilon = F(\rho,\rho') \equiv \| \sqrt{\rho} \sqrt{\rho'} \|_1^2$ is
 the fidelity between the initial and final states, $\braket{E-E_0}$ is the
 expected energy of the state relative to the ground state energy of the
 Hamiltonian, and $\alpha(\epsilon)$ is a function independent of the
 Hamiltonian and the initial state of the system.  They substantiated this
 bound numerically without actually proving it ~\cite{GLM03}.  Researchers
 generally refer to this generalized result also as the \ML bound.

 Giovannetti \emph{et al.} gave a lower and an upper bound of the function
 $\alpha(\epsilon)$ in their paper~\cite{GLM03}.  Specifically, for any $q \ge
 0$, they considered the inequality
\begin{equation}
 \cos x + q\sin x \ge 1 - m x
 \label{E:cos_inequality_GLM}
\end{equation}
 for $x \ge 0$.  Here $m \ge 0$ plus the auxiliary variable $y$ are defined as
 the solution of the system of equations
\begin{subequations}
 \label{E:cos_inequality_conditions_GLM}
\begin{equation}
 m = \frac{y + \sqrt{y^2 (1+q^2) + q^2}}{1+y^2}
\end{equation}
 and
\begin{equation}
 \sin y = \frac{m(1-q y) + q}{1+q^2}
\end{equation}
\end{subequations}
 for $y\in [\pi - \tan^{-1}(1/q),\pi + \tan^{-1}(q)]$.  They then used
 Inequality~\eqref{E:cos_inequality_GLM} to prove that
\begin{align}
 \alpha(\epsilon) &\ge \alpha_{<}(\epsilon) \nonumber \\
 &\equiv \min_\phi \left( \max_q \left\{ \frac{2 [ 1 - \sqrt{\epsilon} (\cos
  \phi - q \sin\phi)]}{\pi m} \right\} \right) .
 \label{E:alpha_lower_bound}
\end{align}
 In addition, by considering the minimum time evolution for the states
 $\ket{\Omega_\xi} = \sqrt{1 - \xi^2} \ket{0} + \xi \ket{E_1}$ (where $0 \le
 \xi \le 1$ and $\ket{0}$ as well as $\ket{E_1}$ are eigenvectors of the
 Hamiltonian with eigenvalues $0$ and $E_1 > 0$, respectively), they further
 showed that
\begin{equation}
 \alpha(\epsilon) \le \alpha_{>}(\epsilon) \equiv \frac{2 z}{\pi} \cos^{-1}
 \left[ 1 - \frac{1-\epsilon}{2 z (1-z)} \right] ,
 \label{E:alpha_upper_bound}
\end{equation}
 where $z$ is the value of $\xi^2$ that minimizes the evolution time.  In
 other words, $z$ is given by
\begin{equation}
 \cos^{-1} \left[ 1 - \frac{1-\epsilon}{2 z (1-z)} \right] = \frac{1-2z}{1-z}
 \sqrt{\frac{1-\epsilon}{\epsilon-1+4z (1-z)}} .
 \label{E:alpha_upper_bound_constraint}
\end{equation}

 Giovannetti \emph{et al.} believed that $\alpha_{<}(\epsilon) =
 \alpha_{>}(\alpha)$ for these two functions agree numerically to at least
 7~significant figures~\cite{GLM03}.  Nonetheless, they failed to give a
 proof.  After almost 20~years, H\"{o}rnedal and S\"{o}nnerborn broke the
 silence on this matter lately.  They showed that $\alpha(\epsilon) =
 \alpha_{>}(\epsilon)$ for qubit systems by explicitly classifying all initial
 qubit states whose evolution times equal the \ML lower bound.  Then they
 extended their prove to higher-dimensional Hilbert space systems and gave a
 symplectic geometry interpretation of the \ML bound~\cite{HS23}.

 Here I show that $\alpha_{>}(\epsilon)$ is indeed equal to
 $\alpha_{<}(\epsilon)$ by first giving an elementary alternative proof of the
 \ML bound.  This proof gives equivalent expressions for
 $\alpha_{>}(\epsilon)$ and $\alpha_{<}(\epsilon)$.  More importantly, it
 makes the necessary and sufficient conditions for saturating the \ML bound
 apparent.  (A pair of initial state and time-independent Hamiltonian is said
 to be saturating the \ML bound if the evolution time $\tau$ equals the
 R.H.S. of Inequality~\eqref{E:ML_bound_original_form}.)  Through these
 conditions, I can write down initial quantum states and their corresponding
 time-independent Hamiltonians that saturates the \ML bound and use them to
 show that $\alpha_{>}(\epsilon) = \alpha_{<}(\epsilon)$.  Finally, I
 investigate the computational aspect of this problem.  I point out that using
 Eq.~\eqref{E:alpha_upper_bound_constraint} to find $\alpha_{>}(\epsilon)$
 can be numerically unstable for $\epsilon$ close to~$1$ and report a simple,
 efficient and accurate way to do so over the entire range of $\epsilon \in
 [0,1]$.

\section{A New Proof Of The \ML Bound}
\label{Sec:ML}

\subsection{Auxiliary Results}
\label{Subsec:lemma}

\begin{lemma}
 Let $\theta \in (-\pi,\pi)$, then
 \begin{equation}
  \cos x \ge \cos\theta - A_\theta (x-\theta) \text{~for all~} x \ge \theta ,
  \label{E:cos_inequality}
 \end{equation}
 where
 \begin{align}
  A_\theta &= \sup_{x > \theta} \frac{\cos\theta - \cos x}{x-\theta} \nonumber
   \\
  &=
  \begin{cases}
   \displaystyle \max_{x \in [\max (\frac{\pi}{2},|\theta|),\pi]} \frac{\cos
    \theta - \cos x}{x-\theta}
   & \text{if~} -\pi < \theta < \frac{\pi}{2} , \\
   \\
   \sin \theta & \text{if~} \frac{\pi}{2} \le \theta < \pi .
  \end{cases}
  \label{E:A_phi_def}
 \end{align}
 Moreover, the supremum in Eq.~\eqref{E:A_phi_def} is attained by a unique
 $x \in [\max(\pi/2,|\theta|),\pi]$.
 \label{Lem:cos_inequality}
\end{lemma}

 From now on, I use the notations $\varphi(\theta)$ or simply $\varphi$ to
 denote the unique $x$ maximizing the second line of Eq.~\eqref{E:A_phi_def}
 when $\theta < \pi/2$.  I also set $\varphi = \theta$ when $\theta \ge
 \pi/2$.

\begin{corollary}
 Inequality~\eqref{E:cos_inequality} can be rewritten as
 \begin{equation}
  \cos x \ge \cos\theta - (x - \theta) \sin\varphi(\theta) \text{~for all~}
  x \ge \theta
  \label{E:cos_inequality_refined}
 \end{equation}
 with equality holds only when $x = \theta$ or $\varphi$.  Moreover, the
 function
 \begin{equation}
  f_\theta(x) = \cos x - \cos\theta + (x-\theta) \sin\varphi(\theta) \ge 0
  \text{~for all~} x \ge \theta .
  \label{E:f_def}
 \end{equation}
 In fact, it has exactly two roots in the interval $[\theta,+\infty)$ provided
 that $\theta \in I_1 \equiv (-\pi,\pi/2)$.  They are a simple root at
 $\theta$ and a double root at $\varphi \in [\max(\pi/2,|\theta|),\pi]$,
 respectively.  Thus, for $\theta \in I_1$, there is a unique $x$ in
 $(\theta,+\infty)$ that maximizes $(\cos\theta - \cos x)/(x-\theta)$ in
 Eq.~\eqref{E:A_phi_def}.  Furthermore, this maximizing $x$ is in
 $[\max(\pi/2,|\theta|),\pi]$.  Whereas for $\theta \in I_2 \equiv
 [\pi/2,\pi)$, $f_\theta$ only has a double root at $\varphi \in
 [\theta,+\infty)$.  Last but not least, $\varphi$ is a decreasing (an
 increasing) function of $\theta \in I_1$ ($\theta \in I_2$), $\varphi -
 \theta$ is a decreasing function of $\theta \in I_1 \cup I_2$ and $\varphi +
 \theta$ is an increasing function of $\theta \in I_1 \cup I_2$.
 \label{Cor:cos_inequality}
\end{corollary}

 Proofs of Lemma~\ref{Lem:cos_inequality} and
 Corollary~\ref{Lem:cos_inequality} can be found in the Appendix.  It is
 natural to apply the above Lemma and Corollary for all values of $\theta \in
 (-\pi,\pi)$ to derive a \QSL in Sec.~\ref{Subsec:ML_bound}.  However, in
 subsequent analysis, I find that only those bounds derived from the case of
 $\theta \in [-\pi/2,0]$ are strong enough to be useful.  More precisely,
 Lemma~\ref{Lem:cos_inequality} and Corollary~\ref{Cor:cos_inequality} with
 $\theta \in [-\pi/2,0]$ can be used to derive the \ML bound.  But no better
 \QSL bound can be obtained by considering $\theta$ outside this interval.

\begin{remark}
 As shown in Fig.~\ref{F:cos_inequality}, the geometric meaning of
 Lemma~\ref{Lem:cos_inequality} is that for $\theta \in (-\pi,\pi)$, the curve
 $y = \cos x$ is always above the line $L$ that meets this cosine curve at no
 more than two points, namely, $(\theta,\cos\theta)$ and $(\varphi,\cos
 \varphi)$ with $\varphi \ge \theta$ whenever $x$ is in the domain $[\theta,
 +\infty)$.  Furthermore, they meet tangentially at the latter point.
 Actually, Lemma~\ref{Lem:cos_inequality} is equivalent to
 Inequality~\eqref{E:cos_inequality_GLM} originally used in
 Refs.~\cite{PHYSCOMP96,ML98,GLM03}.  It is also a generalization of Lemma~$1$
 in Ref.~\cite{Chau10}.  The validity of Corollary~\ref{Cor:cos_inequality} is
 quite evident from Fig.~\ref{F:cos_inequality}.  Note further that
 Inequality~\eqref{E:cos_inequality} in Lemma~\ref{Lem:cos_inequality} is
 valid for $\theta \in (-\pi,\pi)$.  In contrast, Giovannetti \emph{et al.}
 only considered Inequality~\eqref{E:cos_inequality_GLM} with $q \ge 0$ in
 Ref.~\cite{GLM03}, which corresponds to the case of $\theta \in (-\pi,0]$ for
 Inequality~\eqref{E:cos_inequality}.
 \label{Rem:cos_inequality}
\end{remark}

\begin{figure}[t]
 \centering\includegraphics[width=7.5cm]{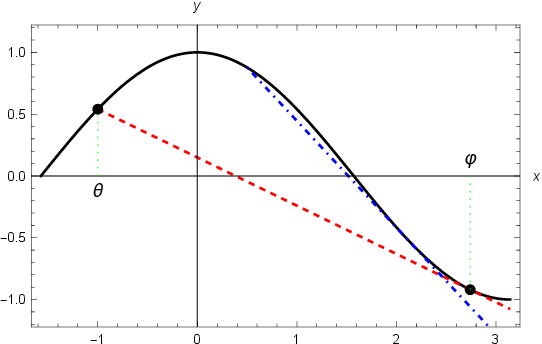}
 \caption{The black curve is $y = \cos x$.  The red dashed and blue
  dash-dotted line segments are the ones corresponding to $\theta = -1.0$ and
  $0.5$, respectively.  In particular, for $\theta = -1.0$, the line segment
  meets tangentially with the cosine curve at $\varphi \approx 2.74$.  From
  this graph, it is evident that as $\theta$ increases from $-\pi$ to $\pi/2$,
  the corresponding $\varphi \in [\pi/2,\pi]$ as well as $\varphi - \theta$
  decrease where as $\varphi + \theta$ increases.  These observations are
  stated in Corollary~\ref{Cor:cos_inequality} and proven in the Appendix.
  \label{F:cos_inequality}}
\end{figure}

\begin{corollary}
 Let $J$ be the interval $[-\pi/2,0]$.  Then
 \begin{subequations}
 \begin{equation}
  \{ \varphi(\theta) \colon \theta \in J \} \subset (\frac{\pi}{2},\pi) ,
  \label{E:varphi_range}
 \end{equation}
 \begin{equation}
  \{ \varphi(\theta) - \theta \colon \theta \in J \} \subset
  (\frac{\pi}{2},\frac{3\pi}{2})
  \label{E:varphi_minus_theta_range}
 \end{equation}
 and
 \begin{equation}
  \{ \varphi(\theta) + \theta \colon \theta \in J \} \subset (0,\pi) .
  \label{E:varphi_plus_theta_range}
 \end{equation}
 \end{subequations}
 \label{Cor:ranges}
\end{corollary}
\begin{proof}
 Eq.~\eqref{E:varphi_range} is a directly consequence of
 Eq.~\eqref{E:A_phi_def} in Lemma~\ref{Lem:cos_inequality}.
 Eqs.~\eqref{E:varphi_minus_theta_range} and~\eqref{E:varphi_plus_theta_range}
 follows from Corollary~\ref{Cor:cos_inequality} and
 Eq.~\eqref{E:varphi_range}.
\end{proof}

\subsection{A New Proof Of The \ML Bound And A New Expression Of
 $\boldsymbol{\alpha_{<}(\epsilon)}$}
\label{Subsec:ML_bound}

 I use the following notations.  Every time-independent Hamiltonian is
 formally written as $\sum_j E_j \ket{E_j} \bra{E_j}$ with $\ket{E_j}$'s being
 the normalized energy eigenstates of the Hamiltonian and $E_0$ being the
 ground state energy.  Furthermore, a normalized initial pure state
 $\ket{\Psi(0)}$ is formally written as $\sum_j a_j \ket{E_j}$.

\begin{theorem}[\ML bound]
 The evolution time $\tau$ needed for any quantum state to evolve to another
 state whose fidelity between them is $\epsilon$ under a time-independent
 Hamiltonian obeys
 \begin{equation}
  \frac{\tau}{\hbar} \ge \max_{\theta\in K_\epsilon} \frac{\cos\theta -
  \sqrt{\epsilon}}{\braket{E-E_0} \sin\varphi(\theta)} \equiv
  \frac{\pi\alpha_{<}(\epsilon)}{2\braket{E-E_0}} ,
  \label{E:ML_bound}
 \end{equation}
 where $K_\epsilon = [-\cos^{-1}(\sqrt{\epsilon}),0]$, $\varphi(\theta)$ is
 the unique root of Eq.~\eqref{E:f_def} in the interval
 $[\max(\pi/2,|\theta|),\pi]$, and $\braket{E-E_0}$ is the expectation value
 of the energy of the system relative to the ground state energy of the
 Hamiltonian.  (Note that denominator of the R.H.S. of
 Inequality~\eqref{E:ML_bound} vanishes if the initial state is a ground state
 of the Hamiltonian.  In this case, Inequality~\eqref{E:ML_bound} still holds
 if one interprets its R.H.S. as $0$ if $\epsilon = 1$ and $+\infty$
 otherwise.)  Last but not least, the $\theta$ maximizing the R.H.S. of
 Inequality~\eqref{E:ML_bound} is unique.
 \label{Thrm:ML_bound}
\end{theorem}
\begin{proof}
 I only need to prove this theorem for pure initial states.  If the initial
 state is mixed, then one just needs to consider the evolution of the purified
 state in the extended Hilbert space~\cite{GLM03}.

 For any fixed $\theta\in (-\pi,\pi)$, using the notations stated at the
 beginning of this Subsection and by Corollary~\ref{Cor:cos_inequality}, I
 obtain
\begin{align}
 \sqrt{\epsilon} &= \left| \braket{\Psi(0) | \Psi(\tau)} \right| \ge
  \Re \left( \braket{\Psi(0) | \Psi(\tau)} e^{i E_0 \tau / \hbar} e^{-i\theta}
  \right) \nonumber \\
 &= \sum_j |a_j|^2 \cos \left[ \frac{(E_j - E_0) \tau}{\hbar} + \theta \right]
  \nonumber \\
 &\ge \sum_j |a_j|^2 \left[ \cos\theta - \frac{(E_j - E_0) \tau \sin
  \varphi(\theta)}{\hbar} \right] \nonumber \\
 &= \cos\theta - \frac{\braket{E-E_0}\tau\sin\varphi(\theta)}{\hbar} .
 \label{E:fidelity1}
\end{align}
 Here $\Re(\cdot)$ denotes the real part of its argument.  Therefore,
\begin{equation}
 \frac{\tau}{\hbar} \ge \sup_{\theta\in (-\pi,\pi)} \frac{\cos\theta -
 \sqrt{\epsilon}}{\braket{E-E_0} \sin\varphi(\theta)}
 \label{E:ML_bound_alt}
\end{equation}
 provided that $\braket{E-E_0} > 0$.  If $\braket{E-E_0} = 0$, $\ket{\Psi(0)}$
 is a ground state of the Hamiltonian.  In this case,
 Inequality~\eqref{E:fidelity1} still holds according to the convention stated
 in this Theorem.  Furthermore, $\braket{E-E_0}$ can never be negative.

 From Eq.~\eqref{E:varphi_range} in Corollary~\ref{Cor:ranges}, $\sin\varphi >
 0$.  Thus, I may exclude those $\theta$'s with $\cos\theta - \sqrt{\epsilon}
 < 0$ from the supremum calculation in the R.H.S. of
 Inequality~\eqref{E:ML_bound_alt}.  That is to say, I only need to consider
 those $\theta \in [-\cos^{-1}(\sqrt{\epsilon}),\cos^{-1}(\sqrt{\epsilon})]
 \subset [-\pi/2,\pi/2]$.  In this domain, Corollary~\ref{Cor:cos_inequality}
 demands $1/\sin\varphi$ to be a decreasing function of $\theta$.  Combined
 with the fact that $\cos\theta - \sqrt{\epsilon}$ is an even function and
 that the R.H.S. of Inequality~\eqref{E:ML_bound_alt} is continuous for
 $\theta\in [-\cos^{-1}(\sqrt{\epsilon}),\cos^{-1}(\sqrt{\epsilon})]$, I
 conclude that the supremum in the R.H.S. of Inequality~\eqref{E:ML_bound_alt}
 can be replaced by maximum over $\theta \in K_\epsilon =
 [-\cos^{-1}(\sqrt{\epsilon}),0]$.  Therefore, I obtain
 Inequality~\eqref{E:ML_bound}.

 Recall from the proof of Corollary~\ref{Cor:cos_inequality} that $\varphi$ is
 a differentiable function of $\theta$.  From Eqs.~\eqref{E:f_def}
 and~\eqref{E:dvarphi_dtheta},
 \begin{align}
  & \frac{d}{d\theta} \left( \frac{\cos\theta - \sqrt{\epsilon}}{\sin\varphi}
   \right) = 0 \nonumber \\
  \Longleftrightarrow{}& (\cos\theta - \sqrt{\epsilon}) (\sin\varphi -
   \sin\theta) = (\cos \varphi - \cos\theta) \sin\theta \nonumber \\
  \Longleftrightarrow{}& \cos\frac{\varphi - \theta}{2} = \sqrt{\epsilon}
   \cos\frac{\varphi + \theta}{2} .
  \label{E:theta_opt_aux}
 \end{align}
 Since $\theta \in K_\epsilon \subset I_1$,
 Corollaries~\ref{Cor:cos_inequality} and~\ref{Cor:ranges} demand that the
 L.H.S. and R.H.S. of the last line of Eq.~\eqref{E:theta_opt_aux} are
 increasing and decreasing functions of $\theta \in I_1$, respectively.
 Therefore, the $\theta \in K_\epsilon$ that maximizes the R.H.S. of
 Inequality~\eqref{E:ML_bound} must be unique.
\end{proof}

\begin{remark}
 The expression of $\alpha_{<}(\epsilon)$ in the R.H.S. of
 Inequality~\eqref{E:ML_bound} is equivalent to that of
 Eq.~\eqref{E:alpha_lower_bound} originally obtained by Giovannetti
 \emph{et al.} in Ref.~\cite{GLM03}.  In fact, they can be transformed from
 one to another via the equations $q = -\tan\theta$ and $\theta = \phi$
 relating the optimized $\theta, \phi$ and $q$.  Besides,
 Eq.~\eqref{E:cos_inequality_GLM} can be simplified as $m \cos\theta = \sin
 \varphi$.  From its proof, it is straightforward to see that the \ML bound in
 Theorem~\ref{Thrm:ML_bound} can slightly strengthen to
 \begin{equation}
  \frac{\tau}{\hbar} \ge \max_{\theta\in K_\epsilon} \frac{\cos\theta -
  \sqrt{\epsilon}}{\braket{E - \underline{E}} \sin\varphi(\theta)}
  \label{E:ML_bound_improved}
 \end{equation}
 where $\underline{E} = \esssup \{ E \colon \sum_{E_j < E} |a_j|^2 = 0 \}$.
 In fact, Inequality~\eqref{E:ML_bound_improved} had been reported in
 Ref.~\cite{HS23}.
 \label{Rem:new_alpha_lower}
\end{remark}

\subsection{$\boldsymbol{\alpha_{>}(\epsilon) = \alpha_{<}(\epsilon)}$}
\label{Subsec:equality_of_alpha}
 From now on, I denote the $\theta$ that maximizes the R.H.S. of
 Inequality~\eqref{E:ML_bound} by $\theta_\text{opt}$.  Moreover, I denote
 $\varphi(\theta_\text{opt})$ by $\varphi_\text{opt}$.

\begin{theorem}
 For each $\epsilon \in [0,1]$, there exists a pair of (pure) quantum state
 and time-independent Hamiltonian saturating the \ML bound in
 Theorem~\ref{Thrm:ML_bound}.  In fact, for $\epsilon = 1$, any quantum state
 and Hamiltonian pair can saturate the \ML bound.  For $\epsilon \in [0,1)$,
 an initial (pure and normalized) quantum state $\ket{\Psi(0)}$ and a
 Hamiltonian pair saturate the \ML bound if and only if
 \begin{equation}
  \ket{\Psi(0)} = a_0 \ket{E_0} + a_1 \ket{E_1}
  \label{E:Psi_optimal_form}
 \end{equation}
 with $E_1 > E_0$, $(E_1-E_0)\tau/\hbar = \varphi_\text{opt} -
 \theta_\text{opt} > 0$ and
 \begin{equation}
  |a_1|^2 = \frac{\sin \left| \theta_\text{opt} \right|}{2 \sin \left(
  \frac{\varphi_\text{opt} - \theta_\text{opt}}{2} \right) \cos \left(
  \frac{\varphi_\text{opt} + \theta_\text{opt}}{2} \right)} .
  \label{E:a1_value}
 \end{equation}
 (Note that here the required time-independent Hamiltonian $H$ appears
 implicitly via its energy eigenstates in Eq.~\eqref{E:Psi_optimal_form}.
 Explicitly, $H = E_0 \ket{E_0}\bra{E_0} + E_1 \ket{E_1}\bra{E_1} + H'$ where
 $H' \ge E_0$ is a time-independent Hamiltonian whose support equals the
 orthogonal complement of the span of $\ket{E_0}$ and $\ket{E_1}$.  Note
 further that although $H'$ does not affect the evolution of $\ket{\Psi(0)}$,
 the requirement $H' \ge E_0$ is essential though technical.  This is because
 $H'$ makes the \ML bound sub-optimal by shifting the ground state energy if
 it has an eigenvalue less than $E_0$.)  Thus, $\alpha(\epsilon) =
 \alpha_{<}(\epsilon)$.
 \label{Thrm:ML_saturation}
\end{theorem}
\begin{proof}
 For $\epsilon = 1$, Inequality~\eqref{E:ML_bound} becomes $\tau \ge 0$.
 Since $\epsilon = F(\ket{\Psi(0)},\ket{\Psi(0)}) = 1$, this inequality is
 just an equality for any initial state $\ket{\Psi(0)}$ under the action of
 any Hamiltonian.

 In the remaining proof, $\epsilon$ is assumed to be in $[0,1)$.  By
 Corollary~\ref{Cor:cos_inequality} and the proof of
 Theorem~\ref{Thrm:ML_bound}, the necessary and sufficient conditions for a
 pair of quantum state $\ket{\Psi(0)}$ and time-independent Hamiltonian to
 saturate the \ML bound in Inequality~\eqref{E:ML_bound} are:
 \begin{enumerate}
  \item \label{Item:theta_opt} the $\theta$ in Inequality~\eqref{E:fidelity1}
   is equal to $\theta_\text{opt} \in K_\epsilon$; and
  \item this $\theta = \theta_\text{opt}$ together with $\ket{\Psi(0)}$ also
   turn the Inequality~\eqref{E:fidelity1} into an equality.
 \end{enumerate}

 Using $|z| = \max_{\theta\in{\mathbb R}} \Re(z e^{-i\theta})$ of all $z
 \in {\mathbb C}$, a trick first used in \QSL research in Ref.~\cite{U92},
 the first line of Inequality~\eqref{E:fidelity1} becomes an equality if and
 only if
 \begin{enumerate}
  \item $\braket{\Psi(0)|\Psi(\tau)} \ne 0$ and $\theta = \theta_\text{opt} =
   \arg(\braket{\Psi(0)|\Psi(\tau)}$ $e^{i E_0 \tau / \hbar}) \bmod 2\pi$; or
  \item $\braket{\Psi(0)|\Psi(\tau)} = 0$ and $\theta = \theta_\text{opt}$ can
   be any real number.
 \end{enumerate}
 Note that $\tau > 0$ as $\epsilon < 1$.  For each term in the second line of
 Inequality~\eqref{E:fidelity1}, I set $x = (E_j - E_0)\tau/\hbar + \theta$
 and apply Corollary~\ref{Cor:cos_inequality} to it.  In this way, I know that
 the third line of Inequality~\eqref{E:fidelity1} becomes an equality if and
 only if $a_j = 0$ whenever $E_j \ge E_0$ and $(E_j - E_0)\tau/\hbar -
 \theta_\text{opt} \notin \{ \theta_\text{opt}, \varphi_\text{opt} \}$.  In
 other words, the (normalized) initial state must be in the form of
 Eq.~\eqref{E:Psi_optimal_form} with $E_1 > E_0$ and $(E_1-E_0)\tau/\hbar =
 \varphi_\text{opt} - \theta_\text{opt} > 0$.  To conclude, for a given
 $\epsilon \in [0,1)$, the initial state $\ket{\Psi(0)}$ saturating the \ML
 bound are the ones given by Eq.~\eqref{E:Psi_optimal_form} with
 $\theta_\text{opt} \in K_\epsilon \subset [-\pi/2,0]$.  And the corresponding
 Hamiltonian is the one with $(E_1 - E_0)\tau/\hbar = \varphi_\text{opt} -
 \theta_\text{opt}$.

 Recall that $\theta_\text{opt}$ must also equal the argument of
 $\braket{\Psi(0)|\Psi(\tau)} e^{i E_0 \tau / \hbar} \bmod 2\pi$ if
 $\braket{\Psi(0)|\Psi(\tau)} \ne 0$.  Thus, $a_1$ in
 Eq.~\eqref{E:Psi_optimal_form} obeys
 \begin{align}
  \displaybreak[1]
  \tan\theta_\text{opt} &= \frac{\Im \left( \braket{\Psi(0)|\Psi(\tau)} e^{i
   E_0 \tau / \hbar} \right)}{\Re \left( \braket{\Psi(0)|\Psi(\tau)} e^{i E_0
   \tau / \hbar} \right)} \nonumber \\
  &= \frac{-|a_1|^2 \sin(\varphi_\text{opt} - \theta_\text{opt})}{|a_0|^2 +
   |a_1|^2 \cos(\varphi_\text{opt} - \theta_\text{opt})} \nonumber \\
  &= \frac{-|a_1|^2 \sin(\varphi_\text{opt} - \theta_\text{opt})}{1 - |a_1|^2
   [1 - \cos(\varphi_\text{opt} - \theta_\text{opt})]}
  \label{E:a1}
 \end{align}
 where $\Im(\cdot)$ is the imaginary part of its argument.  By changing
 $|a_1|^2$ as the subject, I get
 \begin{align}
  |a_1|^2 &= \frac{\tan\theta_\text{opt}}{\tan\theta_\text{opt} [1 - \cos
   (\varphi_\text{opt} - \theta_\text{opt})] - \sin (\varphi_\text{opt} -
   \theta_\text{opt})} \nonumber \\
  &= \frac{\sin \left| \theta_\text{opt} \right|}{2 \sin \left(
   \frac{\varphi_\text{opt} - \theta_\text{opt}}{2} \right) \cos \left(
   \frac{\varphi_\text{opt} + \theta_\text{opt}}{2} \right)} ,
  \label{E:a1_value1}
 \end{align}
 which is Eq.~\eqref{E:a1_value}.  Although Eq.~\eqref{E:a1} is ill-defined
 when its denominator is zero, Eq.~\eqref{E:a1_value} is well-defined and
 correct in all cases.  Specifically, the first case for Eq.~\eqref{E:a1} to
 be ill-defined is that $\theta_\text{opt} = -\pi/2 \bmod 2\pi$.  Then, the
 vanishing denominator of Eq.~\eqref{E:a1} simplifies to $|a_1|^2 = [2\sin^2
 (\varphi_\text{opt}/2 + \pi/4)]^{-1}$.  It is straightforward to check that
 this expression reduces to Eq.~\eqref{E:a1_value}.  The other case of concern
 is when $\sqrt{\epsilon} = \braket{\Psi(0)|\Psi(\tau)} = 0$.  This case
 reduces to $\varphi_\text{opt} - \theta_\text{opt} = \pi$ and $|a_1|^2 = 1/2$
 by Eq.~\eqref{E:varphi_minus_theta_range}.  And from
 Corollary~\ref{Cor:cos_inequality}, I know that $\cos\varphi_\text{opt} =
 \cos\theta_\text{opt} - \pi \sin\varphi_\text{opt}$.  This gives $\tan
 \varphi_\text{opt} = \tan\theta_\text{opt} = -2/\pi$.  Eliminating
 $\varphi_\text{opt}$ from Eq.~\eqref{E:a1_value} and using the fact that
 $\sin\theta_\text{opt} \ne 0$, I obtain $|a_1|^2 = 1/2$.  This concludes that
 Eq.~\eqref{E:a1_value} is valid in all cases.

 Finally, I need to check that $\ket{\Psi(0)}$ is a valid state by proving
 $|a_1|^2 \in [0,1]$.  Here I prove a slightly stronger result that $|a_1|^2
 \in [0,1/2]$.  Eq.~\eqref{E:a1_value} implies that $|a_1|^2 > 0$.
 Furthermore, showing $|a_1|^2 \le 1/2$ is equivalent to proving
 \begin{widetext}
 \begin{align}
  & \quad -\sin\theta_\text{opt} = \sin \left| \theta_\text{opt} \right| \le
   \sin \left( \frac{\varphi_\text{opt} - \theta_\text{opt}}{2} \right) \cos
   \left( \frac{\varphi_\text{opt} + \theta_\text{opt}}{2} \right) =
   \frac{1}{2} ( \sin\varphi_\text{opt} - \sin\theta_\text{opt}) \nonumber \\
  \Longleftrightarrow & \quad \sin\varphi_\text{opt} + \sin\theta_\text{opt} =
   2 \sin \left( \frac{\varphi_\text{opt} + \theta_\text{opt}}{2} \right) \cos
   \left( \frac{\varphi_\text{opt} - \theta_\text{opt}}{2} \right) \ge 0 .
  \label{E:upper_bound_of_a1_squ_condition}
 \end{align}
 \end{widetext}
 From Eqs.~\eqref{E:varphi_minus_theta_range}
 and~\eqref{E:varphi_plus_theta_range} in Corollary~\ref{Cor:ranges}, I
 conclude that the last line of
 Inequality~\eqref{E:upper_bound_of_a1_squ_condition} is true.  In other
 words, $\ket{\Psi(0)}$ is a valid normalized initial quantum state saturating
 the \ML bound.  This completes the proof.
\end{proof}

\begin{remark}
 Theorem~\ref{Thrm:ML_saturation} can be used to derive the following result
 reported in Ref.~\cite{HS23}.  For any $\epsilon \in [0,1]$ and any initial
 normalized pure state $\ket{\Psi(0)}$, there is a time-independent
 Hamiltonian $H$ acting on a Hilbert space of dimension at least~2 such that
 $|\braket{\Psi(0)|\Psi(\tau)}|^2 = \epsilon$ with $\tau$ saturating the \ML
 bound.  Likewise, for any $\epsilon \in [0,1]$ and any time-independent
 Hamiltonian $H$ that is not proportional to the identity operator, there is a
 normalized initial state $\ket{\Psi(0)}$ such that
 $|\braket{\Psi(0)|\Psi(\tau)}|^2 = \epsilon$ with $\tau$ saturating the \ML
 bound.  The proof is simple.  As the two settings are trivially true for
 $\epsilon = 1$, I only need to consider the case of $\epsilon \in [0,1)$.
 For the first setting, once $\epsilon$ is given, $\theta_\text{opt}$ is
 fixed.  One chooses $\tau = 1$ and picks $E_0$ and $E_1$ satisfying the
 constraints in Theorem~\ref{Thrm:ML_saturation}.  Moreover, one selects an
 arbitrary but fixed normalized state $\ket{\Phi}$ orthogonal to
 $\ket{\Psi(0)}$.  Let $\ket{E_0} = a_0 \ket{\Psi(0)} + a_1 \ket{\Phi}$ and
 $\ket{E_1} = a_1 \ket{\Psi(0)} - a_0 \ket{\Phi}$, where $a_j$'s are
 non-negative real numbers with $a_1$ obeying Eq.~\eqref{E:a1_value}.  Then,
 it is easy to check that $\ket{\Phi(0)}$ satisfies
 Eq.~\eqref{E:Psi_optimal_form} and $H = E_0 \ket{E_0} \bra{E_0} + E_1
 \ket{E_1} \bra{E_1}$ is the required time-independent Hamiltonian.  For the
 second setting, since $H$ is not proportional to the identity operator, it
 has at least two distinct eigenenergies, say the ground state energy $E_0$
 and an excited state energy $E_1$.  Surely, for any fixed $\epsilon \in
 [0,1)$, one can find $\tau$ making $E_0, E_1$ to satisfy the constraints in
 Theorem~\ref{Thrm:ML_saturation}.  Clearly, the normalized state in the form
 of Eq.~\eqref{E:Psi_optimal_form} with probability amplitude $a_1$ satisfying
 Eq.~\eqref{E:a1_value} is the required initial pure quantum state.
 \label{Rem:HS_comparison}
\end{remark}

\begin{corollary}
 $\alpha_{>}(\epsilon) = \alpha_{<}(\epsilon)$ where
 \begin{align}
  \alpha_{>}(\epsilon) &= \min_{|a_1|^2 \in
   [\frac{1-\sqrt{\epsilon}}{2},\frac{1}{2}]} \frac{4 |a_1|^2}{\pi} \sin^{-1}
   \sqrt{\frac{1-\epsilon}{4 |a_1|^2 \left( 1 - |a_1|^2 \right)}} \nonumber \\
  &= \min_{\mu \in [\sin^{-1}\sqrt{1-\epsilon},\frac{\pi}{2}]} \frac{2\mu
   \left[ 1 - \sqrt{1 - (1-\epsilon) \csc^2 \mu} \right]}{\pi} .
  \label{E:alpha_upper_expression}
 \end{align}
 Here $|a_1|^2$ and $\mu$ are related by
 \begin{equation}
  \mu = \sin^{-1} \sqrt{\frac{1-\epsilon}{4 |a_1|^2 \left( 1 - |a_1|^2
  \right)}} .
  \label{E:a1_to_mu}
 \end{equation}
 \label{Cor:alpha_lower_equals_alpha_upper}
\end{corollary}
\begin{proof}
 The proof of Theorem~\ref{Thrm:ML_saturation} clearly shows that a state
 $\ket{\Psi(0)}$ saturating the \ML bound must be in the form of
 Eq.~\eqref{E:Psi_optimal_form}.  Besides, the optimality condition depends
 on the magnitude rather than the phase of $a_1$.  Since for a given $\epsilon
 \in [0,1]$, the values of the optimal $\theta = \theta_\text{opt}$ and the
 corresponding $\varphi = \varphi_\text{opt}$ are fixed.  So from
 Eq.~\eqref{E:a1_value}, for a given $\epsilon \in [0,1]$, there is only one
 $|a_1|^2 \in [0,1/2]$ that makes the state $\ket{\Psi(0)}$ saturating the \ML
 bound.  As
 \begin{align}
  \epsilon &= \left| \braket{\Psi(0)|\Psi(\tau)} \right|^2 \nonumber \\
  &= 1 - 4 |a_1|^2 \left( 1 - |a_1|^2 \right) \sin^2 \left[ \frac{(E_1 - E_0)
   \tau}{2\hbar} \right] ,
  \label{E:epsilon_a1_relation}
 \end{align}
 I conclude that
 \begin{equation}
  \frac{\tau}{\hbar} = \min_{|a_1|^2 \in [0,1/2]} \frac{2}{E_1 - E_0}
  \sin^{-1} \sqrt{\frac{1-\epsilon}{4 |a_1|^2 \left( 1 - |a_1|^2 \right)}}
  \label{E:state_saturation_bound}
 \end{equation}
 as long as $1-\epsilon \le 4 |a_1|^2 (1 - |a_1|^2)$ or equivalently
 $(1-\sqrt{\epsilon})/2 \le |a_1|^2 \le (1+\sqrt{\epsilon})/2$.  Note that if
 $1 - \epsilon > 4|a_1|^2 (1 - |a_1|^2)$, Eq.~\eqref{E:epsilon_a1_relation}
 has no real-valued solution for $\tau$.  Therefore, the corresponding value
 of $|a_1|^2$ can be excluded from the $\alpha_{>}(\epsilon)$ calculation.
 Since $\braket{E-E_0} = |a_1|^2 E_1$ for $\ket{\Psi(0)}$, I prove the
 validity of the first equality in Eq.~\eqref{E:alpha_upper_expression}.

 Since Eq.~\eqref{E:a1_to_mu} is a bijection from $|a_1|^2 \in
 [(1-\sqrt{\epsilon})/2,1/2]$ to $\mu \in [\sin^{-1}\sqrt{1-\epsilon},\pi/2]$
 whenever $\epsilon < 1$, the last equality in
 Eq.~\eqref{E:alpha_upper_expression} is correct if $\epsilon \in [0,1)$.
 Finally, $\alpha_{>}(0) = 0$ according to the last line of
 Eq.~\eqref{E:alpha_upper_expression}.  So,
 Eq.~\eqref{E:alpha_upper_expression} is also true when $\epsilon = 1$.
\end{proof}

\begin{remark}
 Actually, the first line of Eq.~\eqref{E:alpha_upper_expression} is equal to
 the expression of $\alpha_{>}(\epsilon)$ originally reported in
 Ref.~\cite{GLM03} and reproduced as Eq.~\eqref{E:alpha_upper_bound} in this
 paper.  To show this fact, I let $z = |a_1|^2$.  Then, my claim is correct
 if
 \begin{equation}
  2 \sin^{-1} \sqrt{\frac{1 - \epsilon}{4 z (1-z)}} = \cos^{-1} \left[ 1 -
  \frac{1-\epsilon}{2z(1-z)} \right] .
  \label{E:remark_equiv}
 \end{equation}
 Note that the values of arc sine and arc cosine in Eq.~\eqref{E:remark_equiv}
 are in the principle branch.  So the correctness of
 Eq.~\eqref{E:remark_equiv} can be proven by taking cosine in both sides of
 this equation and then by using compound angle formula.  Nevertheless, there
 is a slight difference in the region of minimization.  In
 Corollary~\ref{Cor:alpha_lower_equals_alpha_upper}, $|a_1|^2$ is minimized
 over a smaller interval of $[(1-\sqrt{\epsilon})/2,1/2]$, whereas in
 Ref.~\cite{GLM03}, it is minimized over a larger interval of $[0,1]$.
 \label{Rem:alpha_upper_bound}
\end{remark}

\section{Efficient And Reliable Computation Of
 $\boldsymbol{\alpha(\epsilon)}$}
\label{Sec:computational_complexity}

 One could compute $\alpha(\epsilon)$ through $\alpha_{<}(\epsilon)$ in
 Eq.~\eqref{E:ML_bound}.  This method involves two maximizations --- one for
 finding $\varphi$ given $\theta$ and the other for maximizing $(\cos\theta -
 \sqrt{\epsilon})/\sin\varphi$ over $\theta$.  Hence, it is very slow if
 generic optimization methods are used.  (There is a minor point.  For the
 trivial case of $\epsilon = 1$, there is nothing to maximize in
 Eq.~\eqref{E:ML_bound} as $\theta$ is fixed and value of $\varphi$ no longer
 relevant even if one insists on computing $\alpha(0)$ numerically.  I exclude
 this special case in almost all the subsequent discussions.)  Another method
 is to compute the above two maximizations by finding the unique roots of
 $f_\theta(\varphi) = 0$ and the last line of Eq.~\eqref{E:theta_opt_aux},
 respectively.  This is faster.  Nevertheless, I do not discuss the rate of
 convergence and stability of this approach here because I am going to report
 a much better method in the next paragraph.  The third way is to compute
 $\alpha(\epsilon)$ via $\alpha_{>}(\epsilon)$ in
 Eq.~\eqref{E:alpha_upper_expression} using a general minimization algorithm.
 This is faster than the first method as it involves only one minimization
 over $|a_1|^2$ or $\mu$ though it is slower than the second method.
 (Actually, no minimization is required for the special case of $\epsilon = 0$
 as the interval for minimization becomes a point.)  Minimization via $\mu$ is
 preferred as it is numerically more stable.  The only potential trouble is
 the serious rounding error in computing the square root part of the
 expression in the R.H.S. of Eq.~\eqref{E:alpha_upper_expression} when $\mu
 \approx \sin^{-1} \sqrt{1-\epsilon}$ or $\pi/2$.  Fortunately, this lost of
 significance has very little effect on the accuracy of the whole expression
 to be minimized in the second line of Eq.~\eqref{E:alpha_upper_expression}.

 There is one more way to compute $\alpha(\epsilon)$ that could be more
 efficient than a general minimization algorithm that is applicable to a
 smooth target function with possibly multiple local minima.  The trick is to
 use an additional property of the function to be minimized.  From the proofs
 of Theorem~\ref{Thrm:ML_saturation} and
 Corollary~\ref{Cor:alpha_lower_equals_alpha_upper} together with the fact
 that Eq.~\eqref{E:a1_to_mu} is a diffeomorphism if $\epsilon \in [0,1]$,
 there is a unique $\mu \in [\sin^{-1}\sqrt{1-\epsilon},\pi/2]$ minimizing the
 second line of Eq.~\eqref{E:alpha_upper_expression} if $\epsilon \in [0,1)$.
 By differentiating the expression in the second line of
 Eq.~\eqref{E:alpha_upper_expression}, I find that this minimizing $\mu$ obeys
\begin{equation}
 g(\mu) \equiv \frac{1 + (1-\epsilon)(\mu \cot\mu - 1)\csc^2 \mu}{\sqrt{1 -
 (1-\epsilon) \csc^2 \mu}} - 1  = 0 .
 \label{E:minimizing_a1}
\end{equation}
 (For $\epsilon = 1$, Eq.~\eqref{E:minimizing_a1} is trivial, giving no
 constraint on $\mu$.  Nevertheless, substituting any real-valued $\mu$ to
 Eq.~\eqref{E:alpha_upper_expression} still gives the correct answer of
 $\alpha_{>}(0) = 0$ if one insists on computing it numerically.  For
 $\epsilon = 0$, no minimization is needed as $\mu$ must be $\pi/2$.)  In this
 way, this particular minimization problem is reduced to a potentially much
 easier problem of finding a unique simple root in a closed interval of a
 single equation.  (In contrast, the second method requires root finding of
 two coupled equations for Eq.~\eqref{E:theta_opt_aux} depends on the solution
 of $f_\theta(\varphi) = 0$.)  Numerical experiment shows that Newton's method
 converges for any input $\epsilon \in (0,1)$ using the initial guess
 $(\sin^{-1}\sqrt{1-\epsilon} + \pi/2)/2$, namely, the mid-point of the
 possible interval for $\mu$.  The plot of $g(\mu)$ in Fig.~\ref{F:g_plot} for
 various $\epsilon$ strongly suggests that the basin of attraction of Newton's
 method is the whole possible interval for $\mu$.  In addition, rounding and
 truncation errors are not significant in evaluating $g(\mu)$ as well as the
 R.H.S. of Eq.~\eqref{E:alpha_upper_expression}.  Consequently, one can
 accurately find the simple root $\mu \in [\sin^{-1}\sqrt{1-\epsilon},\pi/2]$
 in Eq.~\eqref{E:minimizing_a1} through the quadratically convergent Newton's
 method.  Substituting this root to Eq.~\eqref{E:alpha_upper_expression} gives
 $\alpha(\epsilon)$.  Among the four, this is the fastest method to compute
 $\alpha(\epsilon)$ for $\epsilon \ne 1$.  Surely, one may further speed
 things up by accelerated convergence methods, but this is not the main point
 here.

\begin{figure}[t]
 \centering\includegraphics[width=7.5cm]{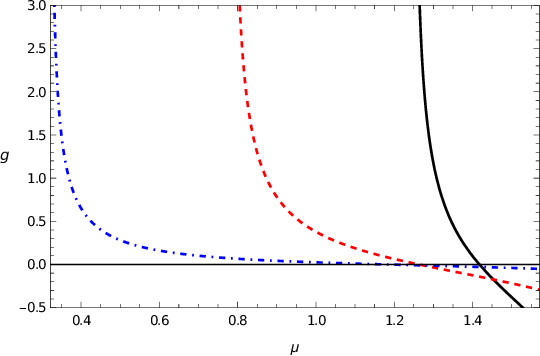}
 \caption{The function $g(\mu)$ for $\epsilon = 0.1$ (black solid curve),
  $0.5$ (red dashed curve) and $0.9$ (blue dot-dashed curve).  It is clear
  from the shape of these curves and their zeros that Newton's method
  converges for almost any initial guess of $\mu \in
  [\sin^{-1}\sqrt{1-\epsilon},\pi/2]$.
  \label{F:g_plot}}
\end{figure}

 Evaluating $\alpha_{>}(\epsilon)$ through numerically finding root of
 Eq.~\eqref{E:alpha_upper_bound_constraint} implicitly suggested in
 Ref.~\cite{GLM03}, in contrast, can be problematic.  It is not clear if
 Eq.~\eqref{E:alpha_upper_bound_constraint} has a unique root in the interval
 of interest although plotting the graph of
 Eq.~\eqref{E:alpha_upper_bound_constraint} strongly suggests it is indeed the
 case.  A more serious problem is numerical instability.  Observe that the
 R.H.S. of Eq.~\eqref{E:alpha_upper_bound_constraint} diverges at $z =
 (1-\sqrt{\epsilon})/2$.  Furthermore, numerical computation shows that the
 root of Eq.~\eqref{E:alpha_upper_bound_constraint} approaches
 $(1-\sqrt{\epsilon})/2$ as $\epsilon \to 1^-$.  In other words, when
 $\epsilon$ is close to $1$, one has to determine the precise location of the
 root of Eq.~\eqref{E:alpha_upper_bound_constraint} close to a singular point.
 To make things worse, the slope of the R.H.S. of
 Eq.~\eqref{E:alpha_upper_bound} diverges at $z = (1-\sqrt{\epsilon})/2$.
 Hence, a highly accurate root $z$ of
 Eq.~\eqref{E:alpha_upper_bound_constraint} is required to evaluate
 $\alpha_{>}(\epsilon)$ via Eq.~\eqref{E:alpha_upper_bound}.  All these can
 only be done with great care.  No wonder why the values of
 $\alpha_{>}(\epsilon)$ computed in this way using Newton's, secant and
 Brent's methods are either inaccurate or divergent when $\epsilon \lesssim
 1$.  For example, using the initial guess $p (1-\sqrt{\epsilon}/2 + (1-p)/2$
 with $p = 1/2$, Newton's method fails to find the root when $\epsilon \gtrsim
 0.76$.  Increasing $p$ to $0.99$, Newton's method works for $\epsilon = 0.76$
 but fails to find the root accurately for $\epsilon \gtrsim 0.9$.  Depending
 on the value of $\epsilon$, great care in choosing initial guess, bounding
 interval and stopping criterion are needed to obtain the root of
 Eq.~\eqref{E:alpha_upper_bound_constraint} and hence the value of
 $\alpha_{>}(\epsilon)$ correctly and accurately.  These complications make
 numerically evaluating $\alpha_{>}(\epsilon)$ by solving
 Eq.~\eqref{E:alpha_upper_bound_constraint} unattractive.

\begin{acknowledgments}
 In memory of my mother Wai Fong Kam.
\end{acknowledgments}

\appendix
\section{Appendix}
\numberwithin{equation}{section}
\begin{proof}[Proof of Lemma~\ref{Lem:cos_inequality}]
 Obviously, Inequality~\eqref{E:cos_inequality} follows directly from the
 first line of Eq.~\eqref{E:A_phi_def}.  So I need to show that the supremum
 in the first line of Eq.~\eqref{E:A_phi_def} exists and is equal to the
 second line of the same equation.  I first consider the case of $\theta \in
 (-\pi,\pi/2)$.  Denote the slope of the line joining the points
 $(\theta,\cos\theta)$ and $(x,\cos x)$ on the cosine curve $y = \cos x$ by
 $M_\theta(x) = (\cos x - \cos\theta)/(x-\theta)$ if $x\ne \theta$ and
 $M_\theta(\theta) = -\sin\theta$ if $x = \theta$.  Clearly, $M_\theta(x) > 0$
 if $x > \theta$ and $\cos x > \cos\theta$.  Moreover, for any fixed $b \in
 [-1,\cos\theta)$, the set $S_{b,\theta} = \{ x > \theta \colon \cos x = b \}$
 is non-empty and $\min S_{b,\theta} \in [|\theta|,\pi]$.  Besides,
 $M_\theta(x) < M_\theta(y) < 0$ whenever $x,y\in S_{b,\theta}$ and $x < y$.
 As a result,
 \begin{align}
  A_\theta &= -\inf_{x>\theta} M_\theta(x) = -\inf_{b\in [-1,\cos\theta)}
   M_\theta(\min S_{b,\theta}) \nonumber \\
  &= -\min_{x\in [|\theta,\pi]} M_\theta(x) ,
  \label{E:A_phi_comp}
 \end{align}
 where the last line is due to continuity of $M_\theta(x)$ and the fact that
 the closure of $\{ \min S_{b,\theta} \colon b\in [-1,\cos\theta) \}$ is
 $[|\theta|,\pi]$.  Therefore, $A_\theta$ is well-defined and the second line
 of Eq.~\eqref{E:A_phi_def} is correct when $\theta \in (-\pi,-\pi/2]$.

 For the case of $\theta \in (-\pi/2,\pi/2)$.  Note that $y = \cos x$ is
 strictly concave in the domain $x\in (-\pi/2,\pi/2)$.  So, Jensen's
 inequality demands that
 \begin{equation}
  p \cos\theta + (1-p) \cos \frac{\pi}{2} < \cos \left[ p\theta +
  \frac{(1-p)\pi}{2} \right]
  \label{E:Jensen}
 \end{equation}
 for $p = (\pi/2 - x)/(\pi/2 - \theta)$ as long as $-\pi/2 <  \theta < x <
 \pi/2$.  Simplifying Eq.~\eqref{E:Jensen} gives $M_\theta(x) >
 M_\theta(\pi/2)$.  Therefore, $\min_{x\in [|\theta|,\pi]} M_\theta(x) =
 \min_{x\in [\pi/2,\pi]} M_\theta(x)$ if $|\theta| < \pi/2$.  So, $A_\phi$ is
 well-defined and Eq.~\eqref{E:A_phi_def} is valid when $\theta \in
 (-\pi/2,\pi/2)$.

 For the case of $\theta \in [\pi/2,\pi)$, $M_\theta(x) > M_\theta(\theta)$
 for all $x \in (\theta,\pi]$ due to strict convexity of $y = \cos x$ in this
 interval.  Thus, $A_\theta = -\lim_{x\to \theta^+} M_\theta(x) = \sin\theta$.
 In other words, Eq.~\eqref{E:A_phi_def} holds if $\pi/2 \le \theta < \pi$.

 Finally, I show that the supremum (and hence maximum when $-\pi < \theta <
 \pi/2$) in Eq.~\eqref{E:A_phi_def} is attained by a unique $x \in
 [\max(\pi/2,|\theta|),\pi]$.  Suppose there were another such $y \ne x$ in
 the same domain that maximizes the second line of Eq.~\eqref{E:A_phi_def}.
 Then the straight line passing through $(\theta,\cos\theta)$, $(x,\cos x)$
 must also pass through $(y,\cos y)$.  Since the cosine function is strictly
 convexity in the interval $(\pi/2,\pi]$, Jensen's inequality implies that
 $\cos x + \cos y > 2 \cos([x+y]/2)$.  Hence, $M_\theta([x+y]/2) <
 M_\theta(x)$, contradicting the assumption that $x$ maximizes
 Eq.~\eqref{E:A_phi_def}.  This completes the proof.
\end{proof}

\begin{proof}[Proof of Corollary~\ref{Cor:cos_inequality}]
 Consider the case of $\theta \in I_1$.  Since the maximum in the second line
 of Eq.~\eqref{E:A_phi_def} is attained at $x = \varphi$, the line $L$ passing
 through $(\theta,\cos\theta)$ and $(\varphi,\cos\varphi)$ must be a tangent
 to the cosine curve $y = \cos x$ at $x = \varphi$.  For the case of $\theta
 \in I_2$, it is clear that the line $L$ and the cosine curve meet
 tangentially at $x = \varphi = \theta$.  Therefore in all cases, $-A_\theta$
 equals the slope of the cosine curve at $x = \varphi$, namely, $-\sin
 \varphi$.  Consequently, Inequality~\eqref{E:cos_inequality} can be rewritten
 as Inequality~\eqref{E:cos_inequality_refined}.  Besides, $f_\theta(x) \ge 0$
 for all $x\ge \theta$.

 Suppose $\theta \in I_1$.  Then the smooth function $f_\theta(x)$ has exactly
 three roots in $[\theta,\pi]$ counted by multiplicity, namely, a simple root
 at $x = \theta$ and a double root at $x=\varphi$.  For otherwise,
 $f_\theta(x)$ has at least four roots in $[\theta,\pi]$.  Hence,
 $f_\theta'(x) = -\sin x + \sin\varphi$ has at least three roots in
 $(\theta,\pi) \subset (-\pi,\pi)$, which is absurd.  Furthermore, using the
 notation and proof of Lemma~\ref{Lem:cos_inequality}, $M_\theta(x) >
 M_\theta(\varphi)$ for all $x > \pi$.  This implies $f_\theta(x) > 0$
 whenever $x > \pi$.  Therefore, $\theta$ and $\varphi$ are the only roots of
 $f_\theta(x)$ in the domain $[\theta,+\infty)$.  In other words, $x =
 \varphi$ is the unique point in $(\theta,+\infty)$ that maximizes
 $(\cos\theta - \cos x)/(x-\theta)$.  Besides, $x \in
 [\max(\pi/2,|\theta|),\pi]$.  The case when $\theta \in I_2$ can be proven in
 the same manner.

 From Lemma~\ref{Lem:cos_inequality}, $\varphi = \theta$ whenever $\theta \in
 I_2$.  Hence, $\varphi$, $\varphi - \theta$ and $\varphi + \theta$ are
 increasing, decreasing and increasing functions of $\theta \in I_2$,
 respectively.

 I now show that $\varphi$ and $\varphi - \theta$ are decreasing functions of
 $\theta \in I_1$.  As already shown in the second paragraph of this proof,
 for each $\theta \in I_1$, $f_\theta(\varphi) = 0$ has a unique solution
 $\varphi \in I_2$.  Moreover, this $\varphi$ is equal to the $x$ that
 maximizes the second line of Eq.~\eqref{E:A_phi_def} in
 Lemma~\ref{Lem:cos_inequality}.  Regarding $f_\theta(\varphi)$ as a function
 of $\theta$ and $\varphi$, then the implicit function theorem implies that
 \begin{equation}
  \frac{d\varphi}{d\theta} = \frac{\sin\varphi - \sin\theta}{(\varphi-\theta)
  \cos\varphi}
  \label{E:dvarphi_dtheta}
 \end{equation}
 provided that the denominator of Eq.~\eqref{E:dvarphi_dtheta}, namely,
 $\partial f_\theta / \partial\varphi$, is non-zero.  This condition is
 satisfied as $\theta \in I_1$ and $\varphi \in I_2$.  Consequently, from
 Eq.~\eqref{E:dvarphi_dtheta}, to prove that $\varphi$ and hence $\varphi -
 \theta$ are decreasing functions of $\theta \in I_1$, I have to show that
 $\sin\varphi \ge \sin\theta$ for $\theta \in I_1$.  Since this inequality is
 trivially true when $-\pi/2 \le \theta \le 0$, I only need to consider the
 remaining case of $\theta \in (0,\pi/2)$.  In this case, it suffices to prove
 that $\varphi \le \pi - \theta$.  It is straightforward to see that the line
 joining $(\theta,\cos\theta)$ and $(\pi-\theta, \cos [\pi-\theta])$ meets the
 cosine curve $y = \cos x$ also at $x = \pi/2$.  Furthermore, convexity of
 this cosine curve in the domain $[\pi/2,\pi]$ implies that this cosine curve
 must lie below this line for $x \in (\pi/2,\pi - \theta)$.  Hence, the unique
 maximum point in the second line of Eq.~\eqref{E:A_phi_def} is attained at $x
 = \varphi < \pi - \theta$.  So, it is proved.

 Finally, I show that $\varphi + \theta$ is an increasing function of $\theta
 \in I_1$.  I claim that $d\varphi/d\theta \ge -1$ for $\theta \in I_1$.  From
 Eq.~\eqref{E:dvarphi_dtheta} and the fact that $\varphi \in I_2$, I obtain
 \begin{align}
  & \quad \frac{d\varphi}{d\theta} \ge -1 \nonumber \\
  \Longleftrightarrow & \quad \sin\varphi (\sin\varphi - \sin\theta) \le
   -(\varphi - \theta) \sin\varphi \cos\varphi \nonumber \\
  \Longleftrightarrow & \quad \sin\varphi (\sin\varphi - \sin\theta) \le
   \cos\varphi (\cos\varphi - \cos\theta) \nonumber \\
  \Longleftrightarrow & \quad \sin \left( \frac{3\varphi + \theta}{2} \right)
   \sin \left( \frac{\varphi - \theta}{2} \right) \ge 0 .
  \label{E:dvarphi_dtheta_inequ_equivalent}
 \end{align}
 Since $\varphi - \theta$ is a decreasing function of $\theta$ in the domain
 $I_1$, $0 < \varphi - \theta < 2\pi$ in the same domain.  Therefore, it
 suffices to prove that $0 \le 3\varphi + \theta \le 2\pi$ for $\theta \in
 I_1$.  As $\theta > -\pi$ and $\varphi > \pi/2$, surely $3\varphi + \theta >
 0$.  Note that $3\varphi + \theta = 4\varphi - (\varphi - \theta)$ is a
 difference of a decreasing and an increasing function of $\theta \in I_1$.
 As a result, $3\varphi + \theta < \left. (3\varphi + \theta) \right|_{\theta
 = \pi/2} = 2\pi$.  This completes the proof.
\end{proof}

\bibliographystyle{apsrev4-2}
\bibliography{qsl1.4}

\end{document}